\documentclass{amsart}
\allowdisplaybreaks

\newtheorem{thm}{Theorem}[section]

\theoremstyle{definition}

\theoremstyle{remark}
\newtheorem{rem}[thm]{Remark}

\numberwithin{equation}{section}

\newcommand{\Rbb}{\mathbb{R}}
\newcommand{\Nbb}{\mathbb{N}}

\newcommand{\Pbb}{\mathbb{P}}
\newcommand{\Qbb}{\mathbb{Q}}
\newcommand{\Ebb}{\mathbb{E}}

\newcommand{\cF}{\mathcal{F}}
\newcommand{\cK}{\mathcal{K}}
\newcommand{\cX}{\mathcal{X}}
\newcommand{\cE}{\mathcal{E}}

\newcommand{\cM}{\mathcal{M}}

\newcommand{\cL}{\mathcal{L}}
\newcommand{\cH}{\mathcal{H}}

\newcommand{\num}{\pi^{\star}}
\newcommand{\numX}{X_{\num}}

\newcommand{\numr}{R_0^{\num}}
\newcommand{\numR}{R^{\num}}
\newcommand{\nums}{S_0^{\num}}

\newcommand{\si}{\sigma}
\newcommand{\m}{\mu}
\newcommand{\ka}{\kappa}
\newcommand{\la}{\lambda}
\newcommand{\La}{\Lambda}

\newcommand{\wt}{\widetilde}
\newcommand{\wb}{\overline}

\DeclareMathOperator*\uplim{\overline{lim}}
\newtheorem{theorem}{Theorem}[section]

\newtheorem{corollary}[theorem]{Corollary}

\usepackage{setspace, color}

\usepackage[
backend=bibtex,
style=alphabetic,
sorting=nty
 ]{biblatex}
\theoremstyle{definition}
\newtheorem{exam}{Example}[section]
\newtheorem*{exam*}{Example}
\usepackage{stmaryrd}

\usepackage{calrsfs}
\DeclareMathAlphabet{\pazocal}{OMS}{zplm}{m}{n}

\begin{document}

\title[Log-optimality with small liability stream]{Log-optimality with small liability stream}%
\author{Michail Anthropelos}
\address{Department of Banking and Financial Management\\
University of Piraeus}
\email{anthropel@unipi.gr}
\thanks{The research project was supported by the Hellenic Foundation for Research and Innovation (H.F.R.I.) under the ``2nd Call for H.F.R.I. Research Projects to support Faculty Members \& Researchers'' (Project Number: 3444/2022, Project Title: ``\textit{Valuation and Optimal Investment of Pension Plans}'').}
\thanks{We would like to thank Scott Robertson, Mihai Sirbu, Gordan {\v Z}itkovi{\' c} and Thaleia Zariphopoulou for their valuable comments and suggestions.}

\author{Constantinos Kardaras}
\address{Department of Statistics \\
London School of Economics}
\email{k.kardaras@lse.ac.uk}

\author{Constantinos Stefanakis}
\address{Department of Banking and Financial Management\\
University of Piraeus}
\email{kstefanakis@unipi.gr}

\begin{abstract}
In an incomplete financial market with general continuous semimartingale dynamics; we model an investor with log-utility preferences who, in addition to an initial capital, receives units of a non-traded endowment process. Using duality techniques, we derive the fourth-order expansion of the primal value function with respect to the units $\epsilon$, held in the non-traded endowment. In turn, this lays the foundation for expanding the optimal wealth process, in this context, up to second order w.r.t. $\epsilon$. The key processes underpinning the aforementioned results are given in terms of Kunita-Watanabe projections, mirroring the case of lower order expansions of similar nature. Both the case of finite and infinite horizons are treated in a unified manner.
\end{abstract}

\maketitle

\vspace{1cm}
\begin{center}
\textit{October, 2025}
\end{center}

\vspace{1.5cm}
\newpage

\section*{Introduction}
\subsection*{Discussion}

A central problem in financial economics involves an investor allocating initial wealth across an array of assets, with the goal of maximizing the expected utility of terminal wealth. This optimal investment problem in continuous-time settings was initially studied by Merton in \cite{M69,M71}, who used dynamic programming techniques to derive a non-linear partial differential equation characterizing the value function. 

A major conceptual advancement came with the development of the theory of equivalent martingale measures by Ross \cite{R76}, Harrison and Kreps \cite{HK79} and Harrison and Pliska \cite{HP81}, which enabled the application of martingale and duality methods to such optimization problems. Under the assumption of market completeness, this duality approach was further developed by Pliska \cite{P86}; Karatzas, Lehoczky, and Shreve \cite{KLS87}; and Cox and Huang \cite{CH89,CH91}. The more intricate case of incomplete markets was addressed in foundational works by He and Pearson \cite{HP91a,HP91b} and by Karatzas, Lehoczky, Shreve, and Xu \cite{KLSX91}. Building on these contributions, Kramkov and Schachermayer \cite{KS99,KS03} established minimal conditions on both the utility function and the financial market under which the core results of the theory remain valid.

In the context of incomplete markets, a natural extension of the problem involves maximizing expected utility when the investor receives an additional exogenous random endowment. Typical examples of this being pension funds. In complete markets, endowments can be perfectly replicated using traded assets, effectively reducing the problem to one with augmented initial wealth and no random endowment. However, as noted among others in \cite{HH07}, real-world markets are typically incomplete, with perfect replication impeded by frictions such as transaction costs, non-traded assets, and portfolio constraints. In such settings, assets are associated with a range of arbitrage-free prices, and the risk of holding them cannot be fully hedged through market trading alone. 

Consequently, analyzing the value function and its solution in an incomplete market setting is undeniably more challenging. Notable contributions addressing this challenge include Cvitani\'{c}, Schachermayer, and Wang \cite{CSW01}, who characterized the optimal terminal wealth in a general semimartingale model via a dual formulation. \cite{KZ03} extends this framework to account for intertemporal consumption. Hugonnier and Kramkov \cite{HK04} treated both the initial capital and the units held in the endowment as optimization variables, hence not requiring the use of finitely additive measures. \cite{OZ09} studies the case of unbounded random endowments and utility functions defined over the entire real line, providing necessary and sufficient conditions for the existence of a solution to the primal problem. \cite{M15} obtains necessary and sufficient conditions in the general framework of an incomplete financial model with a stochastic field utility and intermediate consumption, occurring according to some stochastic clock, while \cite{M17} generalizes the former by also incorporating a random endowment process into the model. Staying within the context of intermediate consumption, \cite{CCFM17} show that the key conclusions of the utility maximization theory hold under the assumptions of No Unbounded Profit with Bounded Risk (NUPBR) and of the finiteness of both primal and dual functions.

\subsection*{Contributions}

In an incomplete financial market with general continuous semimartingale dynamics, we model an investor with log-utility preferences who, in addition to an initial capital, receives units of a non-traded endowment process. As explicit solutions to the associated utility maximization problem are generally unavailable, even for simple model specifications; to address this, we assume that the payoff from the non-traded endowment is small relative to the investor's total wealth. Within this framework, our contributions to the literature are as follows:
\begin{enumerate}
\item \textit{Fourth-order expansion and nearly optimal strategies:}
Using duality techniques, we derive the fourth-order expansion of the primal value function with respect to the number of units $\epsilon$, held in the non-traded endowment. In turn, this lays the foundation for expanding the optimal wealth process, in this context, up to second order w.r.t. $\epsilon$. To the best of our knowledge this is the first result in this direction, extending the work of \cite{KS06}, \cite{KS07} for the case of log-utility. Interestingly enough, the key processes underpinning the aforementioned results are given in terms of Kunita-Watanabe projections, mirroring the case of lower order expansions of similar nature. 

\item \textit{Long-horizon asymptotics:}
Our model also accommodates  the case of ``infinite time horizons". This is a non-trivial addition which allows for the original problem to be understood in ``myopic terms" for investors with distant maturities. The upshot being that one can obtain explicit results in a larger array of models, as the dependence on the horizon is eliminated.
\end{enumerate}

These results have two key implications. First they allow for a better understanding of log-optimal behavior in the incomplete setting, under the presence of a non-traded endowment. This stems from the fact that the first order approximation of the optimal wealth process with non-traded endowment, w.r.t. $\epsilon$, is actually optimal; assuming market completeness. This provides unique insight on how market incompleteness affects asset allocation. Second, considering the case of arbitrarily large time horizons, i.e. the infinite horizon setting, comes with its own merits. Namely it enables analyzing assets which do not have a certain pre-specified maturity. A prominent example of such a situation arises, for example, in the context of a pension fund's liabilities.

\subsection*{Related literature} 

The existing literature on optimal investment is too vast in order to give a complete overview. Instead, we focus on the specific area of utility-based hedging and pricing, which is closely aligned to our work. 

An appealing choice of utility in this context is the exponential one, as it allows for closed-form results in various settings. This is due to its property of separating the value function into components associated with wealth and trading; simplifying the analysis considerably. Prominent works in this context include \cite{H02}, \cite{MZ04}, \cite{GH07}, and \cite{LL12}. These studies leverage a linearization technique—commonly referred to as the Cole-Hopf transformation or distortion power—first introduced in claim valuation by \cite{Z01}, which reduces the resulting nonlinear HJB PDE to a linear form, solvable via standard methods. Further generalizations by \cite{FS08} and \cite{FS10} showed that, even in models with general asset dynamics, the exponential utility-based price admits an explicit expression; although these formulas are often less interpretable. Complementary to these results, \cite{D06} used duality techniques to derive an explicit form for the optimal hedging strategy; with related developments also appearing in \cite{M13}. 

Even within the relatively tractable exponential utility framework, explicit expressions are not always obtainable. For example, in models where the claim depends on the traded asset Sircar and Zariphopoulou \cite{SZ05} derive asymptotic expansions for the utility-indifference price in the context of fast mean-reverting volatility. Henderson and Liang \cite{HL16} consider a multidimensional non-traded asset model subject to intertemporal default risk, and develop a semigroup approximation using splitting techniques.

Given the scarcity of explicit results, various asymptotic approaches have been proposed for pricing and hedging in incomplete markets. Monoyios \cite{M04, M07}, for example, works within a Black-Scholes framework with basis risk and approximates the hedging strategy in powers of $1 - \rho^2$, where $\rho$ denotes the correlation between traded and non-traded assets. In \cite{H02} and \cite{HH02} the authors consider the case of power utility and derive the second-order expansions of the investor's value function with respect to a small position in the contingent claim, thereby approximating both the hedging strategy and reservation price.

These early results of Henderson and Hobson were significantly extended in \cite{KS06, KS07}, where the authors study a general semimartingale framework; under a broad class of utility functions defined on the positive real line. Particularly in \cite{KS06} they derive the second-order expansion of the value function which in turn is used to get a first order approximation for marginal (utility-based) indifference prices and study their qualitative features. A related analysis by Kallsen \cite{K02} studies first order marginal price approximations under local utility maximization.

In \cite{KS07}, using techniques developed in their earlier work, the authors also provide a first-order approximation of the utility-based hedging strategy w.r.t. the units held in the non-traded endowment and demonstrate its relation to quadratic hedging. Similar asymptotic results are found in \cite{M10}, which considers valuation and hedging in the presence of parameter uncertainty under exponential utility and partial information. Therein, the indifference price is approximated up to linear order in the risk aversion parameter via PDE methods.

In the same spirit, \cite{KR11} analyzes utility-based pricing and hedging under exponential utility for a vanishing risk aversion. \cite{KMV14}, focusing on exponential L\'{e}vy models, presents alternative representations of the results in \cite{KS06, KS07} for power utility functions, that avoid the need for a change of numeraire. 

Within the setting of exponential L\'{e}vy processes, \cite{MT16} derives a non-asymptotic approximation for the exponential utility-based indifference price. The approach therein extends the earlier small risk-aversion asymptotics and yields a closed-form approximation by treating the L\'{e}vy model as a perturbation of the classical Black-Scholes framework.

There is also an associated strand of literature which conducts sensitivity analysis of the utility maximization problem w.r.t. other model perturbations. For example \cite{MS19} studies the sensitivity of the expected utility maximization problem in a continuous semimartingale market w.r.t. small changes in the market price of risk. \cite{M20} investigates the behavior of the expected utility maximization problem under small perturbations of the numeraire, while \cite{MS24} study the response of the optimal investment problem to small changes of the stock price dynamics.

\subsection*{Structure of the paper}
This paper is organized as follows: in \S 2 the setup of the model is given; therein dynamics for the financial market, investment opportunities and illiquid asset are explained. In \S 3 we perform the fourth order expansion for the value function $u(\epsilon)$ of the problem of optimal investment with a random endowment. This section lays the foundation for subsequent results. In \S 4 we derive the second order expansion for the optimal wealth $X^{\epsilon}$, i.e. the unique solution of $u(\epsilon)$. Both in \S 3 and in \S 4, the cases of finite and infinite horizon are treated simultaneously. We conclude with a discussion of  derived results as well as a concrete example, presented in \S 5.
\section{Setup}
\subsection{The financial market}
Fix a probability space $(\Omega,\cF,\Pbb)$ equipped with a filtration $\cF(\cdot)$ satisfying the usual conditions, s.t. $\cF(0)$ is trivial a.s. We set:
\begin{equation*}
\cF(\infty):=\sigma\bigg(\bigcup_{t\in\Rbb_{\geq 0}}\cF(t)\bigg)\subseteq \cF,
\end{equation*}
and work within the infinite horizon setting, i.e. $\Rbb_{\geq 0}=[0,\infty)$, since any finite time horizon can be naturally embedded in it. Unless otherwise explicitly stated, identities or inequalities involving random variables are interpreted as being valid almost everywhere under $\Pbb$. Similarly, comparisons between processes will be understood up to indistinguishability. Moreover for $f:\Omega\rightarrow\Rbb$, let $f^{+}:=\max(f,0)$, $f^{-}:=\max(-f,0)$ and for $p\geq 1$, denote the class of Lebesgue $p$-integrable functions in $(\Omega,\cF,\Pbb)$ by $\cL_{p}$. In cases where it is important to underline that the sigma algebra is different than $\cF$, it will be stated explicitly.

Consider an agent that trades in $1+d\in\Nbb_{>0}$ assets; a baseline with stochastic exponential price $\wt{S}_0$, driven by a continuous finite variation process $R_0$ s.t. $R_0(0)=0$ and $d$ risky assets. Discounting by $\wt{S}_0$ we denote the cumulative returns of the risky assets, relative to the baseline, by $R=(R_{i}; 1\leq i\leq d)$ and assume they satisfy the following continuous semimartingale dynamics:
\begin{equation*}
R_{i}=A_{i}+M_{i}, \qquad 1\leq i \leq d.
\end{equation*}
Here, each real-valued component $A_{i}$ of the vector-valued $A=(A_{i}; 1\leq i\leq d)$ is a continuous finite variation process with $A_{i}(0)=0$, whereas each $M_{i}$ in $M=(M_{i}; 1\leq i \leq d)$ is a continuous local martingale with $M_{i}(0)=0$. In turn, the prices of the risky assets $\wt{S}=(\wt{S}_{i}; 1\leq i \leq d)$, discounted by $\wt{S}_0$, satisfy:
\begin{equation*}
\wt{S}_{i}=\wt{S}_{i}(0)\cE(R_{i}), \qquad \wt{S}_{i}(0)>0, \qquad 1\leq i\leq d,
\end{equation*}
where $\cE(\cdot)$ denotes the stochastic exponential.

We introduce the continuous, nondecreasing scalar process:
\begin{equation*}
O:=\sum_{i=1}^{d}\left(\int_0^{\cdot}|dA_{i}(t)|+[M_{i}]\right),
\end{equation*}
with $\int_0^{S}|dA_{i}(t)|$ denoting the total variation of $A_{i}$ on the interval $[0,S]$, for $S>0$. This scalar process $O$ plays the role of an ``operational clock" for the vector semimartingale $R$: all processes $A_{i}$ and $[M_{i},M_{j}]$ for $1\leq i \leq d$ and $1\leq j \leq d$ are absolutely continuous w.r.t. this clock (i.e. the respective induced measures). Hence, there exist predictable processes $a=(a_{i}; 1\leq i\leq d)$ and $c=(c_{ij}; 1\leq i,j\leq d)$, s.t.:
\begin{equation*}
A=\int_0^{\cdot}a(t)dO(t), \qquad [M]=\int_0^{\cdot}c(t)dO(t).
\end{equation*}
By altering it on a $(\Pbb\otimes O)$-null set if necessary, we shall assume throughout that the process $c$ takes values in the space of symmetric, nonnegative-definite, $d\times d$ matrices.
\subsection{Investment opportunities}
In the market $(1,\wt{S})$, an initial capital $x\in\Rbb$ and a choice of a (self-financing) strategy $\theta=(\theta_{i};1\leq i \leq d)$ (assumed to be predictable and $\wt{S}$-vector integrable) result in a \textit{wealth process}:
\begin{equation}\label{eq:wealth}
\wt{X}(\cdot;x,\theta):=x+\int_0^{\cdot}\theta(t)d\wt{S}(t) \ \footnotemark,
\end{equation}
\footnotetext{We'll usually denote $\wt{X}(\cdot;x,\theta)$ by $\wt{X}$ for notational simplicity.}
where the above is understood as a vector stochastic integral. To avoid the so-called doubling strategies, we restrict the above class as follows: a wealth process $\wt{X}$ with initial capital $\wt{X}(0)=x$ is called \textit{admissible} if:
\begin{center}
$\wt{X}>0$.
\end{center}
We denote this subset of wealth processes by $\wt{\cX}(x)$. The union $\cup_{x>0}\wt{\cX}(x)$ is denoted by $\wt{\cX}$.

The market's viability is intimately connected to the following condition, as shown in \cite[Theorem 2.31]{KK21}:
\begin{equation}\label{eq:ass_arbitrage}\tag{A1}
\text{$c$ is non-singular, $(\Pbb\otimes O)$-a.e. and $\int_0^{K}(a(t))^{'}(c(t))^{-1}a(t)dO(t)<\infty$, $\forall K>0$.} \ \footnotemark
\end{equation}
\footnotetext{A bit more generally it is necessary that $a$ is in the range of $c$, $(\Pbb\otimes O)$-a.e. Defining the ``pseudo-inverse" matrix-valued process $c^{\dagger}:=\lim_{m\rightarrow\infty}((c+(1/m)I_{d})^{-2}c)$, where $I_{d}$ is the identity operator on $\Rbb^{d}$; we also require $\int_0^{K}(a(t))^{'}(c(t))^{\dagger}a(t)dO(t)<\infty$, $\forall K>0$.}
In particular, the above condition implies the existence of the supermartingale numeraire; denote it by $\numX:=\cE(\numR)$, where $\numR:=\int_0^{\cdot}(\num(t))^{'}dR(t)$ and $\num:=c^{-1}a$ (refer to \S 2 in \cite{KK21} for further details on this concept). Note that $\numX$ can be used as a new numeraire, under which each $\wt{X}\in\wt{\cX}$ becomes a local martingale. In fact, by considering the auxiliary market:
\begin{equation*}
S:=\left(S_0:=\frac{1}{\numX},S:=\frac{\wt{S}}{\numX}\right),
\end{equation*}
we have for each wealth process $\wt{X}$ in $(1,\wt{S})$:
\begin{equation}\label{eq:wealth_martingale}
\wt{X}/\numX=x+\int_0^{\cdot}\sum_{i=1}^{d}\theta_{i}(t)dS_{i}(t)+\int_0^{\cdot}\left(\wt{X}(t)/\numX(t)-\sum_{i=1}^{d}\theta_{i}(t)S_{i}(t)\right)\frac{dS_0(t)}{S_0(t)},
\end{equation}
where $S$ is a $\Rbb^{1+d}$-valued local martingale. 
\subsection{Illiquid asset}
Besides trading in the financial market, the agent holds $\epsilon\in\Rbb_{\geq 0}$ units of an exogenous, non-traded cumulative stream $\La$ that is absolutely continuous on $\Rbb_{\geq 0}$, given in discounted terms (w.r.t. $\wt{S}_0$) by:
\begin{equation*}
\La=\int_0^{\cdot}\la(t)/\wt{S}_0(t)dt,
\end{equation*}
for a predictable process $\la$. The only regularity assumption on the illiquid asset is that it can be super and subreplicated using the traded asset $\wt{S}$. In other words, we use the following standing assumption:
\begin{equation}\label{eq:ass_illiquid}\tag{A2}
\text{$\{\wt{X}\in\wt{\cX}:|\La(T)|\leq \wt{X}(T)\}\neq\emptyset$,}
\end{equation}
where $T$ is any a.s. finite stopping time.
\begin{rem}
Note that the non-traded asset $\La$, expressed under the numeraire $\numX$, is given by $L:=\La/\numX$. Now, integration by parts yields:
\begin{equation}\label{eq:illiquid_martingale}
L=\int_0^{\cdot}\La(t)dS_0(t)+F,
\end{equation}
$F:=\int_0^{\cdot}S_0(t)d\La(t)$. In turn, we have for a wealth process $\wt{X}$, with initial capital $x$, that:
\begin{equation*}
\wt{X}/\numX-\int_0^{\cdot}\La(t)dS_0(t),
\end{equation*}
is simply \eqref{eq:wealth_martingale} with a shifted position in $S_0$. In other words, when considering the illiquid asset discounted by the supermartingale numeraire, the non-replicable part of $L$ comes solely from $F$.
\end{rem}
\subsection{Utility maximization problem}
Fix any stopping time $T$ s.t. either $T<\infty$ or $T=\infty$, that will serve as the maturity of the agent. We focus on the model $(S,F)$ in the sequel. Therein an initial capital $x\in\Rbb$ and a choice of $\theta=(\theta_{i};1\leq i \leq d)$ (assumed to be predictable and $S$-integrable) result in a wealth process in $S$:
\begin{equation}\label{eq:wealth_martingale_shifted}
X(\cdot;x,\theta):=x+\int_0^{\cdot}\theta(t)dS(t).
\end{equation}
Hence, all processes of the form shown in \eqref{eq:wealth_martingale} are in fact wealth processes in $S$. The subset of admissible processes will be characterized by:
\begin{equation*}
\text{$X>0$ on $\llbracket 0,T\rrbracket :=\{(\omega,t)\in\Omega\times\Rbb_{\geq 0}:t\leq T(\omega)\}$,}
\end{equation*}
which shall be denoted by $\cX(x)$ and its union $\cup_{x>0}\cX(x)$ by $\cX$. In this context, following \cite{DS97}, a process $X\in\cX$ is called \textit{maximal} if for each $X^{'}\in\cX$ s.t. $X^{'}(T)\geq X(T)$ and $X^{'}(0)=X(0)$, we necessarily have $X^{'}=X$. Lastly, a wealth process in $S$, $X$ is called \textit{acceptable} if there exists a maximal process $X^{'}\in\cX$ s.t. $X+X^{'}\in\cX$. Note that when the horizon is infinite, for each $X\in\cX$ we have $X(\infty):=\lim_{t\rightarrow\infty} X(t)$ a.s., by the fact that $X$ is a continuous positive local martingale. Hence for any acceptable process $X$, $\lim_{t\rightarrow\infty}X(t)$ exists a.s. in that case since $X=X^{'}-X^{''}$, where $X^{'}\in\cX$ and $X^{''}$ is maximal.

Assuming:
\begin{equation}\label{eq:ass_illiquid_infinite}\tag{A3}
\begin{aligned}
&\text{$\lim_{t\rightarrow\infty} F(t)$ exists a.s., and} \\
&\text{$\{X\in\cX:|F(T)|\leq X(T)\}\neq\emptyset$,}
\end{aligned}
\end{equation}
we define the following class:
\begin{equation*}
\cX(x,\epsilon):=\{\text{$X$ a wealth process in $S$}:\text{$X$ is acceptable, $X(0)=x$ and $X(T)-\epsilon F(T)>0$}\}.
\end{equation*}
From the definition of acceptable processes, we also deduce that:
\begin{equation*}
\cX(x,0)=\cX(x), \qquad x>0.
\end{equation*}
The set of points $(x,\epsilon)$ where $\cX(x,\epsilon)$ is not empty is a closed convex cone in $\Rbb^2$. Denote its interior by $\cK:=\{(x,\epsilon):\cX(x,\epsilon)\neq\emptyset\}^{\mathrm{o}}$ and note, as shown in \cite{HK04}, that \eqref{eq:ass_illiquid_infinite} implies $(x,0)\in\cK$, $x>0$.
In turn consider:
\begin{equation*}
u(x,\epsilon):=\sup_{X\in\cX(x,\epsilon)}\Ebb[\ln(X(T)-\epsilon F(T))],
\end{equation*}
i.e. the log-optimization problem in $(S,F)$. To avoid the trivial case we need:
\begin{equation*}
\text{We have $u(x,\epsilon)<\infty$, for some $(x,\epsilon)\in\mathcal{K}$,}
\end{equation*}
In fact the above should always hold for the case of log-utility in this context. To see that note $(x,0)\in\mathcal{K}$ for any $x>0$ under \eqref{eq:ass_illiquid_infinite}, as shown in \cite{HK04}. In particular for any such point we have $\cX(x,0)=\cX(x)$. Using the inequality $\ln(x)\leq x-1, \ x>0$ and the fact that $\Ebb[X(T)]\leq x$ for $X\in\cX(x)$ shows the claim. The concavity of $u(x,\epsilon)$ on the open set $\mathcal{K}$ finally gives that $u(x,\epsilon)<\infty$ for all $(x,\epsilon)\in\mathcal{K}$. 

In turn, following \cite{HK04}, since the log-utility satisfies the asymptotic elasticity condition, i.e. for $U(x):=\ln(x)$ we have $\uplim_{x\rightarrow\infty}xU^{'}(x)/U(x)<1$, $S$ trivially satisfies the condition of NFLVR and we also have \eqref{eq:ass_illiquid_infinite}; the solution to $u(x,\epsilon)$ exists and is unique. We focus on:
\begin{equation}\label{eq:value_discounted}
u(\epsilon):=u(1,\epsilon)=\Ebb[\ln(1+X^{\epsilon}(T)-\epsilon F(T))],
\end{equation}
where $X^{\epsilon}(T)$ denotes the shifted counterpart of the solution of $u(1,\epsilon)$ s.t. it starts from zero.
\begin{rem}\label{rem:orig_opt}
Note that working with the optimization problem in $(S,F)$ doesn't come with any loss of generality in this context, since by using the numeraire invariance of log-utility, along with $(\ln(\numX(T)))^{-}\in\cL_1$ as well as \eqref{eq:ass_illiquid} and \eqref{eq:ass_illiquid_infinite} yields that the unique solution to the respective optimization problem of $u(\epsilon)$ in $(1,\wt{S},\La)$, denoted by $\wt{X}^{\epsilon}$, is given as:
\begin{equation*}
\wt{X}^{\epsilon}(T)=\numX(T)\left(1+X^{\epsilon}(T)+\epsilon\int_0^{T}\La(t)dS_0(t)\right),
\end{equation*}
for any a.s. finite stopping time $T$. Even if the "original" utility maximization problem in $(1,\wt{S},\La)$ is not well-defined for potentially infinite $T$, due to it taking infinite value, it is still well-defined for the numeraire-discounted market. Hence it is preferable to work with the latter for long time horizon settings.
\end{rem}

\section{Fourth order expansion of value function w.r.t. $\epsilon$}

For $p\geq 1$ denote the class of (local) martingales s.t. for each $M$ we have $\Ebb[(\wb{M}(T))^{p}]<\infty$ by $\cH_{p}$; where $\wb{M}:=\sup_{t\in[0,\cdot]}|M(t)|$. Using the above, define:
\begin{flalign*}
&\cM_{p}:=\left\{M\in\cH_{p}:\text{$M(0)=0$ and $M=\int_0^{\cdot}\theta(t)dS(t)$}\right\}.
\end{flalign*}
As it is discussed in \cite{KS06} the following assumption is crucial in order to derive the second order expansion of the value function (with a non-traded endowment), and particularly its lower bound:
\begin{equation*}
\text{$\exists x>0$ and $M\in\cM_{2}$ s.t. $|L(T)|\leq x+M(T)$.}
\end{equation*}
In our context, where we go up to fourth order on the value function, the aforementioned is naturally extended to the following stronger version:
\begin{equation}\label{eq:ass_illiquid_stronger}\tag{A3II}
\begin{aligned}
&\text{$\lim_{t\rightarrow\infty} F(t)$ exists a.s., and} \\
&\text{$\exists x>0$ and $M\in\cM_{4}$ s.t. $|F(T)|\leq x+M(T).$}
\end{aligned}
\end{equation}
Moving forward, the key tools for the fourth order expansion of $u(\epsilon)$ are two Kunita-Watanabe (K-W) projections. Denote by $\cF(T)$ the ``stopped" sigma-algebra at $T$, i.e.:
\begin{equation*}
\cF(T)=\left\{A\in\cF(\infty):A\cap\{T\leq t\}\in\cF(t) \ \text{for all $t\geq 0$}\right\},
\end{equation*} 
where the restriction to sets in $\cF(\infty)$ is to take into account the possibility of the stopping time being infinite. Now, as $F$ is progressive (in $\Rbb_{\geq 0}$) and $F(\infty)$ is $\cF(\infty)$-measurable we have that $F(T)$ is $\cF(T)$-measurable. In fact,  when we have an a.s. infinite horizon it holds that $\cF(T)=\cF(\infty)$. In particular, as $F(T)\in\cL_2(\cF(T))$ it admits an orthogonal projection on the space of stochastic integrals w.r.t. $S$ that start at zero. In turn, this gives rise to the following K-W decomposition:
\begin{equation*}
\Ebb[F(T)|\cF(\cdot)]=\Delta+N,
\end{equation*}
where $N$ is strongly orthogonal to $\Delta$ and $N(0)=\Ebb[F(T)]$. Assuming:
\begin{equation}\label{eq:ass_fourth}\tag{A4}
\Ebb[(N(T))^4]<\infty,
\end{equation}
and applying similar reasoning to the above, we have:
\begin{equation*}
\Ebb[(N(T))^2|\cF(\cdot)]=\Gamma+P,
\end{equation*}
where $P$ is strongly orthogonal to $\Gamma$, which is a stochastic integral w.r.t. $S$, and $P(0)=\Ebb[(N(T))^2]$. In fact, by \eqref{eq:ass_illiquid_stronger} and \eqref{eq:ass_fourth} we have $\Delta,N\in\cH_4$ and $\Gamma,P\in\cH_2$.

\begin{theorem}\label{the:fourth_value}
Assume \eqref{eq:ass_arbitrage}, \eqref{eq:ass_illiquid_stronger} and \eqref{eq:ass_fourth}; then we have:
\begin{equation}\label{eq:fourth_value}
u(\epsilon)+\epsilon\Ebb[F(T)]+\frac{\epsilon^2}{2}\Ebb[(N(T))^2]+\frac{\epsilon^3}{3}\Ebb[(N(T))^3]+\frac{\epsilon^4}{4}\Ebb[(N(T))^4]+\epsilon^4\Ebb[(\Gamma(T))^2/2-(N(T))^2\Gamma(T)]=o(\epsilon^4),
\end{equation}
as $\epsilon\rightarrow 0+$.
\end{theorem}
\begin{proof}
We do some general setup that will be useful in deriving  both bounds for the value function $u(\epsilon)$.

Let $x+M$ be the martingale implied by \eqref{eq:ass_illiquid_stronger}. Define:
\begin{equation*}
\epsilon^{L}:=1\land\frac{1}{6x}\land\frac{1}{4|\Ebb[F(T)]|}\land\frac{1}{2(P(0))^{1/2}}.
\end{equation*}
For $\epsilon\in(0,\epsilon^{L})$, consider the stopping times:
\begin{flalign*}
&\tau_{\epsilon}^{\Delta}:=\inf\{t:|\Delta(t)|\geq1/6\epsilon\}, \qquad \tau_{\epsilon}^{\Gamma}:=\inf\{t:|\Gamma(t)|\geq 1/6\epsilon^2\}, \\
&\tau_{\epsilon}^{N}:=\inf\{t:|N(t)|\geq 1/4\epsilon\}, \qquad \tau_{\epsilon}^{M}:=\inf\{t:x+M(t)\geq 1/6\epsilon\}, \\
&\tau_{\epsilon}^{P}:=\inf\{t:|P(t)|\geq 1/4\epsilon^2\},
\end{flalign*}
and $\tau_{\epsilon}:=\tau_{\epsilon}^{\Delta}\land \tau_{\epsilon}^{\Gamma}\land \tau_{\epsilon}^{N}\land \tau_{\epsilon}^{M}\land\tau_{\epsilon}^{P}$. Now noting that for a non-negative random variable $\zeta$ with $\Ebb[\zeta^{p}]<\infty$ and $p>0$ implies:
\begin{equation*}
\lim_{z\rightarrow\infty}z^{p}\Pbb(\zeta>z)=0,
\end{equation*}
we have:
\begin{flalign*}
&\lim_{\epsilon\rightarrow 0+}\frac{1}{\epsilon^4}\Pbb(\tau_{\epsilon}^{\Delta}\leq T)=0, \qquad
\lim_{\epsilon\rightarrow 0+}\frac{1}{\epsilon^4}\Pbb(\tau_{\epsilon}^{\Gamma}\leq T)=0, \\
&\lim_{\epsilon\rightarrow 0+}\frac{1}{\epsilon^4}\Pbb(\tau_{\epsilon}^{N}\leq T)=0, \qquad
\lim_{\epsilon\rightarrow 0+}\frac{1}{\epsilon^4}\Pbb(\tau_{\epsilon}^{M}\leq T)=0, \\
&\lim_{\epsilon\rightarrow 0+}\frac{1}{\epsilon^4}\Pbb(\tau_{\epsilon}^{P}\leq T)=0,
\end{flalign*}
since $\Delta,N,M\in\cH_{4}$ and $\Gamma,P\in\cH_2$. In particular the above also give:
\begin{equation}\label{eq:stopping_zero}
\lim_{\epsilon\rightarrow 0+}\frac{1}{\epsilon^4}\Pbb(\tau_{\epsilon}\leq T)=0.
\end{equation}
\underline{\textbf{The lower bound}}

Now note that:
\begin{equation*}
1+\epsilon (\Delta^{\tau_{\epsilon}}+M-M^{\tau_{\epsilon}})+\epsilon^2 \Gamma^{\tau_{\epsilon}}\in \cX(1,\epsilon).
\end{equation*}
In turn, Taylor expansion and the above give:
\begin{flalign*}
u(\epsilon)&\geq \Ebb\left[\ln\left(1+\epsilon\left(\Delta^{\tau_{\epsilon}}(T)+M(T)-M^{\tau_{\epsilon}}(T)-F(T)\right)+\epsilon^2 \Gamma^{\tau_{\epsilon}}(T)\right)\right] \\
&=-\epsilon\Ebb[F(T)]-\frac{1}{2}\Ebb\left[\left(\epsilon\left(\Delta^{\tau_{\epsilon}}(T)+M(T)-M^{\tau_{\epsilon}}(T)-F(T)\right)+\epsilon^2 \Gamma^{\tau_{\epsilon}}(T)\right)^2\right] \\
&+\frac{1}{3}\Ebb\left[\left(\epsilon\left(\Delta^{\tau_{\epsilon}}(T)+M(T)-M^{\tau_{\epsilon}}(T)-F(T)\right)+\epsilon^2 \Gamma^{\tau_{\epsilon}}(T)\right)^3\right] \\
&-\frac{1}{4}\Ebb\left[\left(\epsilon\left(\Delta^{\tau_{\epsilon}}(T)+M(T)-M^{\tau_{\epsilon}}(T)-F(T)\right)+\epsilon^2 \Gamma^{\tau_{\epsilon}}(T)\right)^4\xi_{\epsilon}^{L}\right],
\end{flalign*}
where $\xi_{\epsilon}^{L}$ is a r.v., s.t. $\lim_{\epsilon\rightarrow 0+}\xi_{\epsilon}^{L}=1$ a.s. and $0<\xi_{\epsilon}^{L}\leq 2^4$, $\forall\epsilon\in(0,\epsilon^{L})$. Hence, as it is bounded above, we can disregard it in the following estimates since it does not affect limits being zero as $\epsilon\rightarrow 0+$.

Now define $Q^{\epsilon}:= \Delta^{\tau_{\epsilon}}+M-M^{\tau_{\epsilon}}\in\cM_4$, where $Q^{\epsilon}(T)\rightarrow \Delta(T)$ almost surely as $\epsilon\rightarrow 0+$ (similarly $\Gamma^{\tau_{\epsilon}}(T)\rightarrow \Gamma(T)$ a.s.). Collecting terms, we get:
\begin{flalign*}
u(\epsilon) &\geq -\epsilon\Ebb[F(T)]-\epsilon^2\Ebb[(Q^{\epsilon}(T)-F(T))^2/2]+\epsilon^3\Ebb[(Q^{\epsilon}(T)-F(T))^3/3-\Gamma^{\tau_{\epsilon}}(T)(Q^{\epsilon}(T)-F(T))] \\
&-\epsilon^4\Ebb[(Q^{\epsilon}(T)-F(T))^4\xi_{\epsilon}^{L}/4-(Q^{\epsilon}(T)-F(T))^2\Gamma^{\tau_{\epsilon}}(T)+(\Gamma^{\tau_{\epsilon}}(T))^2/2] \\
&+\epsilon^5\Ebb[(Q^{\epsilon}(T)-F(T))(\Gamma^{\tau_{\epsilon}}(T))^2-(Q^{\epsilon}(T)-F(T))^3\Gamma^{\tau_{\epsilon}}(T)\xi_{\epsilon}^{L}] \\
&+\epsilon^6\Ebb[(\Gamma^{\tau_{\epsilon}}(T))^3/3-(3/2)(Q^{\epsilon}(T)-F(T))^2(\Gamma^{\tau_{\epsilon}})^2\xi_{\epsilon}^{L}]-\epsilon^7\Ebb[(Q^{\epsilon}(T)-F(T))(\Gamma^{\tau_{\epsilon}}(T))^3\xi_{\epsilon}^{L}] \\
&-\epsilon^8\Ebb[(\Gamma^{\tau_{\epsilon}}(T))^4\xi_{\epsilon}^{L}/4].
\end{flalign*}
At this point we use the following result: let $\zeta$ be a non-negative r.v. and assume that $\Ebb[\zeta^{p_1}]<\infty$ holds for some $p_1>0$. The, for all $p_2>0$:
\begin{equation}\label{eq:expectation_tail}
\lim_{z\rightarrow\infty}\frac{1}{z^{p_2}}\Ebb[\zeta^{p_1+p_2}\mathbf{1}_{\zeta\leq z}]=0.
\end{equation}
We examine the above terms (w.r.t. to $\epsilon$) one by one, beginning with the highest order.

\underline{Eighth order term}
\begin{flalign*}
\epsilon^4(\Gamma^{\tau_{\epsilon}}(T))^4&=\epsilon^4(\Gamma^{\tau_{\epsilon}}(T))^4\mathbf{1}_{\tau_{\epsilon}\leq T}+\epsilon^4(\Gamma^{\tau_{\epsilon}}(T))^4\mathbf{1}_{\tau_{\epsilon}=\infty} \\
&\leq \frac{1}{(6\epsilon)^4}\mathbf{1}_{\tau_{\epsilon}\leq T}+\epsilon^4(\overline{\Gamma}(T))^4\mathbf{1}_{\overline{\Gamma}(T)\leq 1/6\epsilon^2}.
\end{flalign*}
Hence:
\begin{flalign*}
\epsilon^4\Ebb[(\Gamma^{\tau_{\epsilon}}(T))^4]&\leq \frac{1}{(6\epsilon^4)}\Pbb(\tau_{\epsilon}\leq T)+\epsilon^4\Ebb\left[(\overline{\Gamma}(T))^4\mathbf{1}_{\overline{\Gamma}(T)\leq 1/6\epsilon^2}\right].
\end{flalign*}
In turn, $\lim_{\epsilon\rightarrow 0+}\epsilon^4\Ebb[(\Gamma^{\tau_{\epsilon}}(T))^4]=0$ follows by \eqref{eq:stopping_zero}, $\Gamma\in\cM_2$ and \eqref{eq:expectation_tail} after a simple change of variables.

\underline{Seventh order term}

Holder's inequality gives:
\begin{flalign*}
\epsilon^3\Ebb[|Q^{\epsilon}(T)-F(T)||\Gamma^{\tau_{\epsilon}}(T)|^3]&\leq \bigg(\Ebb\bigg[\bigg(\bigg|M(T)-F(T)\bigg|+\sup_{t\in[0,T]}\bigg|\Delta(t)-M(t)\bigg|\bigg)^4\bigg]\bigg)^{\frac{1}{4}}\left(\Ebb\left[(\epsilon \Gamma^{\tau_{\epsilon}}(T))^4\right]\right)^{\frac{3}{4}},
\end{flalign*}
where the first expectation is finite by \eqref{eq:ass_illiquid_stronger}, \eqref{eq:ass_fourth} and:
\begin{equation*}
\lim_{\epsilon\rightarrow 0+}\epsilon^4\Ebb[(\Gamma^{\tau_{\epsilon}}(T))^4]=0,
\end{equation*}
by the previous result on the eight order term.

\underline{Sixth order terms}

The first part follows similarly to the eight order term, since:
\begin{equation*}
\epsilon^2|\Gamma^{\tau_{\epsilon}}(T)|^3\leq \frac{1}{6^3\epsilon^4}\mathbf{1}_{\tau_{\epsilon}\leq T}+\epsilon^2(\overline{\Gamma}(T))^3\mathbf{1}_{\overline{\Gamma}(T)\leq 1/6\epsilon^2}.
\end{equation*}
Hence, as before, $\lim_{\epsilon\rightarrow 0+}\epsilon^2\Ebb[|\Gamma^{\tau_{\epsilon}}(T)|^3]=0$ follows by \eqref{eq:stopping_zero}, $\Gamma\in\cM_2$ and \eqref{eq:expectation_tail}.

For the second part note that by Holder's inequality we have:
\begin{flalign*}
\epsilon^2\Ebb[(Q^{\epsilon}(T)-F(T))^2(\Gamma^{\tau_{\epsilon}}(T))^2]&\leq \bigg(\Ebb\bigg[\bigg(\bigg|M(T)-F(T)\bigg|\\
&+\sup_{t\in[0,T]}\bigg|\Delta(t)-M(t)\bigg|\bigg)^4\bigg]\bigg)^{\frac{1}{2}}\left(\Ebb\left[(\epsilon \Gamma^{\tau_{\epsilon}}(T))^4\right]\right)^{\frac{1}{2}},
\end{flalign*}
where the first expectation is finite by \eqref{eq:ass_illiquid_stronger}, \eqref{eq:ass_fourth} and:
\begin{equation*}
\lim_{\epsilon\rightarrow 0+}\epsilon^4\Ebb[(\Gamma^{\tau_{\epsilon}}(T))^4]=0,
\end{equation*}
by the previous result on the eight order term.

\underline{Fifth order terms}

For the first part we have that:
\begin{flalign*}
\epsilon\Ebb[(Q^{\epsilon}(T)-F(T))(\Gamma^{\tau_{\epsilon}}(T))^2]&=\Ebb\left[\Ebb\left[(Q^{\epsilon}(T)-F(T))(\Gamma^{\tau_{\epsilon}}(T))^2\Big|\cF(T\land\tau_{\epsilon})\right]\right] \\
&=-\Ebb[\epsilon N^{\tau_{\epsilon}}(T)(\Gamma^{\tau_{\epsilon}}(T))^2].
\end{flalign*}
Now note that:
\begin{flalign*}
\Ebb[|\epsilon N^{\tau_{\epsilon}}(T)|(\Gamma^{\tau_{\epsilon}}(T))^2]&\leq\Ebb[|\epsilon N^{\tau_{\epsilon}}(T)|(\overline{\Gamma}(T))^2],
\end{flalign*}
where $4|\epsilon N^{\tau_{\epsilon}}(T)|\leq 1$. In turn, since $\lim_{\epsilon\rightarrow 0+}|\epsilon N^{\tau_{\epsilon}}(T)|=0$ a.s., we have:
\begin{equation*}
\lim_{\epsilon\rightarrow 0+}\epsilon\Ebb[(Q^{\epsilon}(T)-F(T))(\Gamma^{\tau_{\epsilon}}(T))^2]=0.
\end{equation*}

For the second part,
\begin{equation*}
\lim_{\epsilon\rightarrow 0+}\epsilon\Ebb[|Q^{\epsilon}(T)-F(T)|^3|\Gamma^{\tau_{\epsilon}}(T)|]=0,
\end{equation*}
follows from Holder's inequality and $\lim_{\epsilon\rightarrow 0+}\epsilon^4\Ebb[(\Gamma^{\tau_{\epsilon}}(T))^4]=0$, similarly to the sixth and seventh order terms.

\underline{Fourth order terms}

We have:
\begin{flalign*}
\lim_{\epsilon\rightarrow 0+}\Ebb[(\Gamma^{\tau_{\epsilon}}(T))^2]&=\Ebb[(\Gamma(T))^2], \\
\lim_{\epsilon\rightarrow 0+}\Ebb[(Q^{\epsilon}(T)-F(T))^2\Gamma^{\tau_{\epsilon}}(T)]&=\Ebb[(N(T))^2\Gamma(T)], \\
\lim_{\epsilon\rightarrow 0+}\Ebb[(Q^{\epsilon}(T)-F(T))^4\xi_{\epsilon}^{L}]&=\Ebb[(N(T))^4],
\end{flalign*}
which all hold by dominated convergence, \eqref{eq:ass_illiquid_stronger}, \eqref{eq:ass_fourth}, $0<\xi_{\epsilon}^{L}\leq 2^4$ and $\lim_{\epsilon\rightarrow 0+}\xi_{\epsilon}^{L}=1$ a.s.

\underline{Third order terms}

Initially note that:
\begin{flalign*}
\Ebb[\Gamma^{\tau_{\epsilon}}(T)(Q^{\epsilon}(T)-F(T))]&=\Ebb\left[\Ebb\left[\Gamma^{\tau_{\epsilon}}(T)(Q^{\epsilon}(T)-F(T))\Big|\cF(T\land\tau_{\epsilon})\right]\right]=0.
\end{flalign*}

For the second part we have by Holder's inequality:
\begin{flalign*}
\frac{1}{\epsilon}\Ebb[|(Q^{\epsilon}(T)-F(T))^3+(N(T))^3|]&=\frac{1}{\epsilon}\Ebb[|(Q^{\epsilon}(T)-F(T))^3+(N(T))^3|\mathbf{1}_{\tau_{\epsilon}\leq T}] \\
&\leq 2^{\frac{1}{4}}\bigg(\Ebb\bigg[\bigg(\bigg|M(T)-F(T)\bigg|\\
&+\sup_{t\in[0,T]}\bigg|\Delta(t)-M(t)\bigg|\bigg)^4+(N(T))^4\bigg]\bigg)^{\frac{3}{4}}\\
&\cdot\left(\frac{1}{\epsilon^4}\Pbb(\tau_{\epsilon}\leq T)\right)^{\frac{1}{4}},
\end{flalign*}
where the expectation on the last inequality is finite by \eqref{eq:ass_illiquid_stronger}, \eqref{eq:ass_fourth} and $\lim_{\epsilon\rightarrow 0+}(\epsilon^4)^{-1}\Pbb(\tau_{\epsilon}\leq T)=0$ by \eqref{eq:stopping_zero}. Hence the second part also tends to zero.

\underline{Second order term}

We have:
\begin{equation*}
\lim_{\epsilon\rightarrow 0+}\frac{1}{\epsilon^2}\Ebb[|(Q^{\epsilon}(T)-F(T))^2-(N(T))^2|]=0,
\end{equation*}
following the same process as for the third order term.

Lastly, combining all the above we have
\begin{flalign}\label{eq:fourth_value_lower}
&\uplim_{\epsilon\rightarrow 0+}\frac{1}{\epsilon^4}\Big(u(\epsilon)+\epsilon\Ebb[F(T)]+\frac{\epsilon^2}{2}\Ebb[(N(T))^2]+\frac{\epsilon^3}{3}\Ebb[(N(T))^3]+\frac{\epsilon^4}{4}\Ebb[(N(T))^4]\\
&+\epsilon^4\Ebb[(\Gamma(T))^2/2-(N(T))^2\Gamma(T)]\Big) \geq 0. \notag
\end{flalign}
\underline{\textbf{The upper bound}}

Note that by Legendre-Fenchel duality we have:
\begin{flalign*}
\ln(1+X^{\epsilon}(T)-\epsilon F(T)) &\leq -1 -\ln(1+\epsilon N^{\tau_{\epsilon}}(T)+\epsilon^2 P^{\tau_{\epsilon}}(T))\\
&+(1+\epsilon N^{\tau_{\epsilon}}(T)+\epsilon^2 P^{\tau_{\epsilon}}(T))(1+X^{\epsilon}(T)-\epsilon F(T)).
\end{flalign*}
Now, denoting the probability measure induced by $(1+\epsilon N^{\tau_{\epsilon}}(T)+\epsilon^2 P^{\tau_{\epsilon}}(T))/\Ebb[1+\epsilon N^{\tau_{\epsilon}}(T)+\epsilon^2 P^{\tau_{\epsilon}}(T)]$ as $\Qbb^{\epsilon}$ we claim that $1+X^{\epsilon}$ is a supermartingale under $\Qbb^{\epsilon}$ and therefore:
\begin{equation*}
\Ebb^{\Qbb^{\epsilon}}[1+X^{\epsilon}(T)]\leq 1.
\end{equation*}
Indeed, we have $1+X^{\epsilon}(T)-\epsilon F(T)>0$ and in particular from \eqref{eq:ass_illiquid_stronger} we get $\epsilon |F(T)|\leq \epsilon(x+M(T))$. It follows that $1+X^{\epsilon}+\epsilon(x+M)$ is a positive stochastic integral w.r.t. $S$. This, along with the fact that $N$, $P$ are strongly orthogonal to $S$ yields that $1+X^{\epsilon}+\epsilon(x+M)$ is a supermartingale under $\Qbb^{\epsilon}$. The claim now follows from the fact that $\epsilon(x+M)$ is a martingale under $\Qbb^{\epsilon}$. Hence we get:
\begin{flalign*}
\Ebb[(1+\epsilon N^{\tau_{\epsilon}}(T)+\epsilon^2 P^{\tau_{\epsilon}}(T))(1+X^{\epsilon}(T))]&=\Ebb[1+\epsilon N^{\tau_{\epsilon}}(T)+\epsilon^2 P^{\tau_{\epsilon}}(T)]\Ebb^{\Qbb^{\epsilon}}[1+X^{\epsilon}(T)]\\
&\leq \Ebb[1+\epsilon N^{\tau_{\epsilon}}(T)+\epsilon^2 P^{\tau_{\epsilon}}(T)].
\end{flalign*}
Furthermore we have:
\begin{flalign*}
&\Ebb[N^{\tau_{\epsilon}}(T)F(T)]=\Ebb[N^{\tau_{\epsilon}}(T)N(T)]=\Ebb[(N^{\tau_{\epsilon}}(T))^2], \\
&\Ebb[P^{\tau_{\epsilon}}(T)F(T)]=\Ebb[P^{\tau_{\epsilon}}(T)N(T)]=\Ebb[P^{\tau_{\epsilon}}(T)N^{\tau_{\epsilon}}(T)],
\end{flalign*}
given the definitions and properties of $N$, $P$. It follows that:
\begin{flalign*}
u(\epsilon)&\leq -\epsilon\Ebb[F(T)]-\epsilon^2\Ebb[(N^{\tau_{\epsilon}}(T))^2]-\epsilon^3\Ebb[P^{\tau_{\epsilon}}(T)N^{\tau_{\epsilon}}(T)]\\
&+\Ebb[-\ln(1+\epsilon N^{\tau_{\epsilon}}(T)+\epsilon^2 P^{\tau_{\epsilon}}(T))+\epsilon N^{\tau_{\epsilon}}(T)+\epsilon^2 P^{\tau_{\epsilon}}(T)].
\end{flalign*}
In turn, Taylor expansion yields:
\begin{flalign*}
\Ebb[-\ln(1+\epsilon N^{\tau_{\epsilon}}(T)+\epsilon^2 P^{\tau_{\epsilon}}(T))+\epsilon N^{\tau_{\epsilon}}(T)+\epsilon^2 P^{\tau_{\epsilon}}(T)]&=\frac{1}{2}\epsilon^2\Ebb[(N^{\tau_{\epsilon}}(T)+\epsilon P^{\tau_{\epsilon}}(T))^2]\\
&-\frac{1}{3}\epsilon^3\Ebb[(N^{\tau_{\epsilon}}(T)+\epsilon P^{\tau_{\epsilon}}(T))^3] \\
&+\frac{1}{4}\epsilon^4\Ebb[(N^{\tau_{\epsilon}}(T)+\epsilon P^{\tau_{\epsilon}}(T))^4\xi_{\epsilon}^{U}],
\end{flalign*}
where the random variables $\xi_{\epsilon}^{U}$ satisfy:
\begin{equation*}
0<\xi_{\epsilon}^{U}\leq 2^4, \ \forall\epsilon\in(0,\epsilon^{L}); \qquad \text{$\lim_{\epsilon\rightarrow 0+}\xi_{\epsilon}^{U}=1$ a.s.}
\end{equation*}
Note that we can once more disregard $\xi_{\epsilon}^{U}$ from the estimates, similarly to $\xi_{\epsilon}^{L}$, since it is bounded above by $2^4$. Collecting terms of same order, we obtain:
\begin{flalign*}
\Ebb[-\ln(1+\epsilon N^{\tau_{\epsilon}}(T)+\epsilon^2 P^{\tau_{\epsilon}}(T))+\epsilon N^{\tau_{\epsilon}}(T)+\epsilon^2 P^{\tau_{\epsilon}}(T)]&=\epsilon^2 \Ebb\left[\frac{(N^{\tau_{\epsilon}}(T))^2}{2}\right] \\
&+\epsilon^3\Ebb\left[N^{\tau_{\epsilon}}(T)P^{\tau_{\epsilon}}(T)-\frac{(N^{\tau_{\epsilon}}(T))^3}{3}\right] \\
&+\epsilon^4\Ebb\bigg[\frac{(P^{\tau_{\epsilon}}(T))^2}{2}-(N^{\tau_{\epsilon}}(T))^2P^{\tau_{\epsilon}}(T)\\
&+\frac{(N^{\tau_{\epsilon}}(T))^4}{4}\xi_{\epsilon}^{U}\bigg]\\
&+\epsilon^5\Ebb[-N^{\tau_{\epsilon}}(T)(P^{\tau_{\epsilon}}(T))^2\\
&+(N^{\tau_{\epsilon}}(T))^3P^{\tau_{\epsilon}}(T)\xi_{\epsilon}^{U}] \\
&\epsilon^6\Ebb\left[-\frac{(P^{\tau_{\epsilon}}(T))^3}{3}+\frac{3}{2}(P^{\tau_{\epsilon}}(T))^2(N^{\tau_{\epsilon}}(T))^2\xi_{\epsilon}^{U}\right] \\
&+\epsilon^7\Ebb[N^{\tau_{\epsilon}}(T)(P^{\tau_{\epsilon}}(T))^3\xi_{\epsilon}^{U}]\\
&+\epsilon^8\Ebb\left[\frac{(P^{\tau_{\epsilon}}(T))^4}{4}\xi_{\epsilon}^{U}\right]
\end{flalign*}
We examine the above terms (w.r.t. to $\epsilon$) one by one, beginning with the highest order.

\underline{Eighth order term}

We claim that:
\begin{equation*}
\lim_{\rightarrow 0+}\epsilon^4\Ebb[(P^{\tau_{\epsilon}}(T))^4]=0.
\end{equation*}
To see this note that since $4\epsilon^2|P^{\tau_{\epsilon}}(T)|\leq 1$:
\begin{equation*}
\epsilon^4(P^{\tau_{\epsilon}}(T))^4\leq \frac{1}{(4\epsilon)^4}\mathbf{1}_{\tau_{\epsilon}\leq T}+\epsilon^4(\overline{P}(T))^4\mathbf{1}_{\overline{P}(T)\leq (4\epsilon^2)^{-1}}.
\end{equation*}
Since $\lim_{\epsilon\rightarrow 0+}\epsilon^{-4}\Pbb(\tau_{\epsilon}<\infty)=0$ by \eqref{eq:stopping_zero}, we only have to show:
\begin{equation*}
\lim_{z\rightarrow\infty}\frac{1}{z^2}\Ebb[(\overline{P}(T))^4\mathbf{1}_{\overline{P}(T)\leq z}]=0,
\end{equation*}
which follows from \eqref{eq:expectation_tail}, similarly to the way we handled the eighth order term in the lower bound of the value function.

\underline{Seventh order term}

We show that:
\begin{equation*}
\lim_{\epsilon\rightarrow 0+}\epsilon^3\Ebb[|N^{\tau_{\epsilon}}(T)||P^{\tau_{\epsilon}}(T)|^3]=0.
\end{equation*}
This follows from Holder's inequality as:
\begin{equation*}
\epsilon^3\Ebb[|N^{\tau_{\epsilon}}(T)||P^{\tau_{\epsilon}}(T)|^3]\leq (\Ebb[(\overline{N}(T))^4]^{\frac{1}{4}}(\Ebb[\epsilon^4(P^{\tau_{\epsilon}}(T))^4])^{\frac{3}{4}},
\end{equation*}
and using \eqref{eq:ass_fourth} as well as $\lim_{\epsilon\rightarrow 0+}\epsilon^{-4}\Pbb(\tau_{\epsilon}<\infty)=0$.

\underline{Sixth order terms}

For the first part we have:
\begin{equation*}
\lim_{\epsilon\rightarrow 0+}\epsilon^2\Ebb[|P^{\tau_{\epsilon}}(T)|^3]=0,
\end{equation*}
which follows similarly to the eighth order term of the upper bound since:
\begin{equation*}
\epsilon^2|P^{\tau_{\epsilon}}(T)|^3\leq \frac{1}{4^3\epsilon^4}\mathbf{1}_{\tau_{\epsilon\leq T}}+\epsilon^2(\overline{P}(T))^3\mathbf{1}_{\overline{P}(T)\leq (4\epsilon^2)^{-1}}.
\end{equation*}
We also have that:
\begin{equation*}
\lim_{\epsilon\rightarrow 0+}\epsilon^2\Ebb[(N^{\tau_{\epsilon}}(T))^2(P^{\tau_{\epsilon}}(T))^2]=0.
\end{equation*}
Indeed, this follows from dominated convergence, as:
\begin{equation*}
\epsilon^2\Ebb[(N^{\tau_{\epsilon}}(T))^2(P^{\tau_{\epsilon}}(T))^2]\leq\Ebb[(\epsilon N^{\tau_{\epsilon}}(T))^2(\overline{P}(T))^2,
\end{equation*}
and $4|\epsilon N^{\tau_{\epsilon}}(T)|\leq 1$, $N\in\cH_4$ (implying $P\in\cH_2$) as well as $\lim_{\epsilon\rightarrow 0+}|\epsilon N^{\tau_{\epsilon}}(T)|=0$.

\underline{Fifth order terms}

For the first part we have that:
\begin{equation*}
\lim_{\epsilon\rightarrow 0+}\epsilon\Ebb[|N^{\tau_{\epsilon}}(T)|(P^{\tau_{\epsilon}}(T))^2]=0,
\end{equation*}
follows similarly to the sixth order terms of the upper bound, while:
\begin{equation*}
\lim_{\epsilon\rightarrow 0+}\epsilon\Ebb[|N^{\tau_{\epsilon}}(T)|^3|P^{\tau_{\epsilon}}(T)|]=0,
\end{equation*}
follows from Holder's inequality since:
\begin{equation*}
\epsilon\Ebb[|N^{\tau_{\epsilon}}(T)|^3|P^{\tau_{\epsilon}}(T)|]\leq \Ebb[(\overline{N}(T))^4]^{\frac{3}{4}}\left(\epsilon^4\Ebb[(P^{\tau_{\epsilon}}(T))^4]\right)^{\frac{1}{4}},
\end{equation*}
and using the fact that $\lim_{\rightarrow 0+}\epsilon^4\Ebb[(P^{\tau_{\epsilon}}(T))^4]=0$.

\underline{Fourth order terms}

We have that:
\begin{flalign*}
&\lim_{\epsilon\rightarrow 0+}\Ebb[(P^{\tau_{\epsilon}}(T))^2]=\Ebb[(P(T))^2], \\
&\lim_{\epsilon\rightarrow 0+}\Ebb[\xi_{\epsilon}^{U}(N^{\tau_{\epsilon}}(T))^4]=\Ebb[(N(T))^4], \\
&\lim_{\epsilon\rightarrow 0+}\Ebb[(N^{\tau_{\epsilon}}(T))^2P^{\tau_{\epsilon}}(T)]=\Ebb[(N(T))^2P(T)],
\end{flalign*}
which all hold by dominated convergence, \eqref{eq:ass_illiquid_stronger}, \eqref{eq:ass_fourth}, $0<\xi_{\epsilon}^{U}\leq 2^4$ and $\lim_{\epsilon\rightarrow 0+}\xi_{\epsilon}^{U}=1$ a.s.

\underline{Third order term}

Holder's inequality gives:
\begin{equation*}
\frac{1}{\epsilon}\Ebb[|(N^{\tau_{\epsilon}}(T))^3-(N(T))^3|]\leq 2(\Ebb[(\overline{N}(T))^4])^{\frac{3}{4}}\left(\frac{1}{\epsilon^4}\Pbb(\tau_{\epsilon}\leq T)\right)^{\frac{1}{4}},
\end{equation*}
which goes to zero by \eqref{eq:stopping_zero}.

\underline{Second order term}

We have:
\begin{equation*}
\frac{1}{\epsilon^2}\Ebb[|(N^{\tau_{\epsilon}}(T))^2-(N(T))^2|]\leq 2(\Ebb[(\overline{N}(T))^4])^{\frac{1}{2}}\left(\frac{1}{\epsilon^4}\Pbb(\tau_{\epsilon}\leq T)\right)^{\frac{1}{2}},
\end{equation*}
which goes to zero, similarly to the third order case for the upper bound.

Combining all the above and after some algebra we have
\begin{flalign}\label{eq:fourth_value_upper}
&\uplim_{\epsilon\rightarrow 0+}\frac{1}{\epsilon^4}\Big(u(\epsilon)+\epsilon\Ebb[F(T)]+\frac{\epsilon^2}{2}\Ebb[(N(T))^2]+\frac{\epsilon^3}{3}\Ebb[(N(T))^3]+\frac{\epsilon^4}{4}\Ebb[(N(T))^4]\\
&+\epsilon^4\Ebb[(\Gamma(T))^2/2-(N(T))^2\Gamma(T)]\Big) \leq 0. \notag
\end{flalign}
Lastly, noting that \eqref{eq:fourth_value_lower} and \eqref{eq:fourth_value_upper} also hold for limit inferior as $\epsilon\rightarrow 0+$ concludes the proof.
\end{proof}

\section{Second order expansion of optimal wealth w.r.t. $\epsilon$}

\begin{theorem}\label{the:second_wealth}
Assume the same conditions as in Theorem \ref{the:fourth_value}; then:
\begin{equation}\label{eq:second_wealth}
 X^{\epsilon}(T)-\epsilon\Delta(T)-\epsilon^2\Gamma(T)=o_{\Pbb}(\epsilon^2) \ \footnotemark[3], 
\end{equation}
as $\epsilon\rightarrow 0+$.
\end{theorem}
\begin{proof}
Take any sequence $\epsilon_{n}\in(0,\epsilon^{L})$ s.t. $\lim_{n\rightarrow \infty}\epsilon_{n}=0$ \footnote{This should come without loss of generality for our purposes since any positive sequence $\epsilon_{n}$, that tends to zero, eventually lies in $(0,\epsilon^{L})$.}. In turn, recalling the process we were considering for the fourth order lower bound of the value function, i.e. $Q^{\epsilon_{n}}$ and Taylor expanding $U(1+\epsilon_{n} Q^{\epsilon_{n}}(T)+\epsilon_{n}^2\Gamma^{\tau_{\epsilon_{n}}}(T)-\epsilon_{n}F(T))$ around $X^{\epsilon_{n}}(T)-\epsilon_{n} F(T)$ we get:
\begin{equation}\label{eq:taylor}
\begin{aligned}
&\ln(1+\epsilon_{n} Q^{\epsilon_{n}}(T)+\epsilon_{n}^2\Gamma^{\tau_{\epsilon_{n}}}(T)-\epsilon_{n}F(T))-\ln(1+X^{\epsilon_{n}}(T)-\epsilon_{n} F(T))= \\
&\frac{1}{1+X^{\epsilon_{n}}(T)-\epsilon_{n} F(T)}\left(\left(1+\epsilon_{n} Q^{\epsilon_{n}}(T)+\epsilon_{n}^2\Gamma^{\tau_{\epsilon_{n}}}(T)-\epsilon_{n}F(T)\right)-\left(1+X^{T,\epsilon_{n}}(T)-\epsilon_{n} F(T)\right)\right) \\
&-\frac{1}{2(1+\xi_{n})^2}\left(X^{\epsilon_{n}}(T)-\epsilon_{n}Q^{\epsilon_{n}}(T)-\epsilon_{n}^2 \Gamma^{\tau_{\epsilon_{n}}}(T)\right)^2,
\end{aligned}
\end{equation}
where $\xi_{n}$ is a r.v. between $X^{\epsilon_{n}}(T)-\epsilon_{n} F(T)$ and $\epsilon_{n} Q^{\epsilon_{n}}(T)+\epsilon_{n}^2\Gamma^{\tau_{\epsilon_{n}}}(T)-\epsilon_{n}F(T)$ s.t. $\lim_{n\rightarrow\infty}\xi_{n}=0$ in probability. This holds by the fact that $\lim_{n\rightarrow\infty}X^{\epsilon_{n}}(T)=0$ in probability (see \cite[Lemma 3.6]{KS99}, \cite[Theorem 1]{KS06} for similar results that cover the case of log-utility). Passing to any subsequence of $\epsilon_{n}$ we can always find a further subsequence, denoted by $\epsilon_{m}$, s.t. $\lim_{m\rightarrow\infty}X^{\epsilon_{m}}(T)=0$ a.s., which in turn implies that the r.v. $\sup_{m}X^{\epsilon_{m}}(T)$ is well-defined a.s. \footnote{In fact note that if \eqref{eq:second_wealth} holds for the subsequence $\epsilon_{m}$, then it also holds for the original (arbitrary) sequence $\epsilon_{n}$ by the double subsequence trick.}.

Now note that:
\begin{equation*}
0<1+\xi_{m}\leq\left(2+\sup_{m}X^{\epsilon_{m}}(T)+|F(T)|\right)\lor\left(2+\overline{Q}^{\epsilon}(T)+\overline{\Gamma}(T)+|F(T)|\right)=:B,
\end{equation*}
which implies that $-1/2(1+\xi_{m})^2\leq -1/2B^2$ and in particular we have $1/B \leq 1$. Hence \eqref{eq:taylor} becomes:
\begin{equation*}
\begin{aligned}
&\ln(1+\epsilon_{m} Q^{\epsilon_{m}}(T)+\epsilon_{m}^2\Gamma^{\tau_{\epsilon_{m}}}(T)-\epsilon_{m}F(T))-\ln(1+X^{\epsilon_{m}}(T)-\epsilon_{m} F(T))\leq \\
&\frac{1}{1+X^{\epsilon_{m}}(T)-\epsilon_{m} F(T)}\left(\left(1+\epsilon_{m} Q^{\epsilon_{m}}(T)+\epsilon_{m}^2\Gamma^{\tau_{\epsilon_{m}}}(T)-\epsilon_{m}F(T)\right)-\left(1+X^{\epsilon_{m}}(T)-\epsilon_{m} F(T)\right)\right) \\
&-\frac{1}{2B^2}\left(X^{\epsilon_{m}}(T)-\epsilon_{m}Q^{\epsilon_{m}}(T)-\epsilon_{m}^2 \Gamma^{\tau_{\epsilon_{m}}}(T)\right)^2.
\end{aligned}
\end{equation*}
Now we claim that:
\begin{equation*}
\Ebb\left[\frac{1+\epsilon_{m} Q^{\epsilon_{m}}(T)+\epsilon_{m}^2\Gamma^{\tau_{\epsilon_{m}}}(T)-\epsilon_{m}F(T)}{1+X^{\epsilon_{m}}(T)-\epsilon_{m} F(T)}\right]\leq 1.
\end{equation*}
This holds since $1+\epsilon_{m} Q^{\epsilon_{m}}+\epsilon_{m}^2\Gamma^{\tau_{\epsilon_{m}}}\in\cX(1,\epsilon)$. Hence:
\begin{equation*}
\Ebb\left[\ln\left(\frac{1+\epsilon_{m} Q^{\epsilon_{m}}(T)+\epsilon_{m}^2\Gamma^{\tau_{\epsilon_{m}}}(T)-\epsilon_{m}F(T)}{1+X^{\epsilon_{m}}(T)-\epsilon_{m} F(T)}\right)\right]\leq 0.
\end{equation*}
Using Jensen's inequality, the claim follows. In turn we have:
\begin{equation}\label{eq:taylor_value}
\frac{u(\epsilon_{m})-\Ebb[\ln(1+\epsilon_{m} Q^{\epsilon_{m}}(T)+\epsilon_{m}^2\Gamma^{\tau_{\epsilon_{m}}}(T)-\epsilon_{m}F(T))]}{\epsilon_{m}^4}\geq \frac{\Ebb[(X^{\epsilon_{m}}(T)-\epsilon_{m}Q^{\epsilon_{m}}(T)-\epsilon_{m}^2 \Gamma^{\tau_{\epsilon_{m}}}(T))^2/2B^2]}{\epsilon_{m}^4},
\end{equation}
where by the fourth order expansion of the value function we have that the left-hand side of \eqref{eq:taylor_value} tends to zero as $m\rightarrow\infty$. Furthermore, we have:
\begin{flalign*}
\frac{\Ebb[(X^{\epsilon_{m}}(T)-\epsilon_{m}\Delta(T)-\epsilon_{m}^2 \Gamma(T))^2/2B^2]}{\epsilon_{m}^4}&\leq \frac{2\Ebb[(X^{\epsilon_{m}}(T)-\epsilon_{m}Q^{\epsilon_{m}}(T)-\epsilon_{m}^2 \Gamma^{\tau_{\epsilon_{m}}}(T))^2/2B^2]}{\epsilon_{m}^4} \\
&+\frac{2\Ebb[(\epsilon_{m}Q^{\epsilon_{m}}(T)+\epsilon_{m}^2 \Gamma^{\tau_{\epsilon_{m}}}(T)-\epsilon_{m}\Delta(T)-\epsilon_{m}^2\Gamma(T))^2]}{\epsilon_{m}^4} \\
&\leq \frac{2\Ebb[(X^{\epsilon_{m}}(T)-\epsilon_{m}Q^{\epsilon_{m}}(T)-\epsilon_{m}^2 \Gamma^{\tau_{\epsilon_{m}}}(T))^2/2B^2]}{\epsilon_{m}^4} \\
&+4\Ebb[(\Gamma^{\tau_{\epsilon_{m}}}(T)-\Gamma(T))^4] \\
&+C\left(\frac{\Pbb(\tau_{\epsilon_{m}}\leq T)}{\epsilon_{m}^4}\right)^{\frac{1}{2}},
\end{flalign*}
for $C>0$. Now the above tends to zero as $m\rightarrow\infty$ by \eqref{eq:taylor_value}, dominated convergence, \eqref{eq:ass_illiquid_stronger}, \eqref{eq:ass_fourth} and \eqref{eq:stopping_zero}. Lastly, Markov's inequality applied on the increasing, positive non-negative function $x\mapsto x^2/2B^2$ gives the desired result.
\end{proof}

\section{Discussion and an example}
Before considering a concrete example, let us comment on the usefulness of the results derived in the previous sections. In particular, besides the obvious increase of accuracy in $X^{\epsilon}(T)$'s approximation by also considering $\Gamma(T)$; there is another arguably more important upside. One that provides some added intuition on the way the investor behaves optimally in an incomplete v.s. complete market setting. In order to underline this we claim that $X^{\epsilon}(T)=\epsilon \Delta(T)$ in a complete market setting, i.e. the first order approximation discussed in \S 2-3 is actually optimal. Indeed, in this case we have that $\Delta(T)=F(T)-\Ebb[F(T)]$, since $N=\Ebb[F(T)]$. This implies that $1+\epsilon\Delta(T)-\epsilon F(T)=1-\epsilon\Ebb[F(T)]$. But then, we get for sufficiently small $\epsilon$ s.t. $1-\epsilon\Ebb[F(T)]>0$ and $\forall X\in\cX(1,\epsilon)$:
\begin{equation*}
\Ebb\left[\frac{X(T)-1-X^{\epsilon}(T)}{1+X^{\epsilon}(T)-\epsilon F(T)}\right]=\Ebb\left[\frac{X(T)-1-\epsilon \Delta(T)}{1-\epsilon\Ebb[F(T)]}\right]=\frac{\Ebb[X(T)]-1}{1-\epsilon\Ebb[F(T)]}\leq 0,
\end{equation*}
which shows that $X^{\epsilon}(T)=\epsilon\Delta(T)$ in view of the first order conditions. Hence, the term $\epsilon^2\Gamma(T)$ in the approximation of $X^{\epsilon}(T)$ allows us to get a sense of how market incompleteness directly affects optimal behavior in a context where a non-traded endowment is present.

Secondly, the fact that the model accommodates infinite horizon markets is a non-trivial addition. In order to better grasp the added benefit of such a result note that while we have a characterization of $\Delta$, $\Gamma$; their dependence of the maturity makes explicit results not possible in many cases. On the other hand the ``myopic case" of $T=\infty$ provides a more tractable tool to approximate optimal behavior. Informally, the message being that in the case of distant horizons asset allocation can be further approximated by the projections of $\Ebb[F(\infty)|\cF(\cdot)]$ and $\Ebb[(N(\infty))^2|\cF(\cdot)]$ on the space of $S$-integrals. Concretely, let us consider a space that supports a two dimensional standard Brownian motion $(W^1,W^2)$, where the risky asset is driven by $\int_0^{\cdot}\si(t)dW^1(t)$, and factor process $Z$:
\begin{equation*}
Z=Z(0)+\int_0^{\cdot}\m(Z(t))dt+\int_0^{\cdot}\ka(Z(t))dB(t),
\end{equation*}
where $\beta:=\ka^2$ and $B:=\rho W^1+\sqrt{1-\rho^2}W^2, \ \rho\in(-1,1)$. With a slight abuse of notation, we assume that all $R_0$, $a$, $\si$, $\la$ are functions of the factor ($c:=\si^2$). Explicit calculations regarding $\Delta^{\infty}$ require us to have a sense of the following form:
\begin{equation*}
\Ebb\left[\int_0^{\infty}\nums(s)\la(s)/S_0(s)ds \bigg|\cF(t)\right]=\int_0^{t}\nums(s)\la(s)/S_0(s)ds+\nums(t)\psi(Z(t))/S_0(t),
\end{equation*}
where $\psi$ satisfies the following second order ODE:
\begin{equation*}
\la-R_0\psi+(\m-\num\si\ka\rho)\partial_z\psi+\frac{1}{2}\beta\partial_{zz}\psi=0.
\end{equation*}
On the contrary, similar calculations for finite horizons would require us to solve a PDE, since a time variable would also be present.

\begin{exam}[Linear payoff on a factor model and an infinite horizon]
In a market with one riskless and one risky asset as well as $T=\infty$, we consider a standard two dimensional Brownian motion $(W^1,W^2)$ generating the augmented filtration $\cF(\cdot)$, and formulate the following problem for $B:=\rho W^1+\sqrt{1-\rho^2}W^2, \ \rho\in(-1,1)$:
\begin{flalign*}
&\La=\int_0^{\cdot}e^{-rt}Z(t)dt, \qquad r>0, \\
&dZ(t)=k(\theta-Z(t))dt+bdB(t), \qquad k,b>0; \ Z(0),\theta\in\Rbb, \\
&dR(t)=adt+\si dW^1(t), \qquad a\in\Rbb; \ \si>0, \\
&d\wt{S}(t)=\wt{S}(t)dR(t).
\end{flalign*}
In turn, in the above model, the portfolio that gives rise to the supermartingale numeraire is constant and in particular we have:
\begin{flalign*}
&\frac{dS_0(t)}{S_0(t)}=\overbrace{-\num\si dW^1(t)}^{d\numr(t)}, \\
&F=\int_0^{\cdot}e^{-rt}S_0(t)Z(t)dt.
\end{flalign*}
Note that even if a non-traded (local) endowment $\la$ given as a linear contract on the factor $Z$ isn't sub/super replicable, the continuity (and adaptedness) of the factor process give that $Z$ is locally bounded. Hence the above are still relevant. In fact define $\tau_{n}:=\inf\{t:|Z(t)|\geq n\}$ (for big enough $n$ s.t. it dominates $Z(0)$). Then $F_{n}:=\int_0^{\cdot}e^{-rt}S_0(t)Z^{\tau_{n}}(t)dt$ satisfies the sub/super replicability conditions for all $n$, since $Z^{\tau_{n}}$ is (uniformly) bounded. Denote the martingales closed by $F(\infty)$, $F_{n}(\infty)$ as $M$ and $M_{n}$ respectively and recall that we have the following Kunita-Watanabe decompositions, as in \S 2:
\begin{flalign*}
M&=\Delta+N, \\
Q&:=\Ebb[(N(\infty))^2|\cF(\cdot)]=\Gamma+P, \\
M_{n}&=\Delta_{n}+N_{n}, \\
Q_{n}&:=\Ebb[(N_{n}(\infty))^2|\cF(\cdot)]=\Gamma_{n}+P_{n},
\end{flalign*}
where $\Delta(0)=\Gamma(0)=\Delta_{n}(0)=\Gamma_{n}(0)=0$. For sufficiently big $r$, we get:
\begin{equation*}
\|[M_{n}-M](\infty)\|_{\cL_1}\rightarrow 0,
\end{equation*}
as $n\rightarrow\infty$.
\end{exam}
Markov's inequality gives that the above convergence also holds in probability. Now, using $[M_{n}-M]=[\Delta_{n}-\Delta]+[N_{n}-N]$ implies that $[\Delta_{n}-\Delta](\infty)\rightarrow 0$ in probability as $n\rightarrow\infty$. In turn, we get:
\begin{equation*}
\sup_{t\geq 0}|\Delta_{n}(t)-\Delta(t)|\rightarrow 0,
\end{equation*}
in probability as $n\rightarrow \infty$. For the case of $\Gamma_{n}$, $\Gamma$ note that we have for any $a>0$:
\begin{equation*}
\Pbb\left(\sup_{t\geq 0}|Q_{n}(t)-Q(t)|>a\right)\leq \frac{1}{a}\|(N_{n}(\infty))^2-(N(\infty))^2\|_{\cL_1(\Pbb)},
\end{equation*}
which holds by Doob's maximal inequality, the fact that $|\Ebb[(N_{n}(\infty))^2-(N(\infty))^2|\cF(\cdot)]|\leq \Ebb[|(N_{n}(\infty))^2-(N(\infty))^2||\cF(\cdot)]$ and $\cup_{K>0}\{\omega:\sup_{0\leq t\leq K}|Q_{n}(t)-Q(t)|>a\}=\{\omega:\sup_{t\geq 0}|Q_{n}(t)-Q(t)|>a\}$. In turn, we get:
\begin{flalign*}
\Pbb\left(\sup_{t\geq 0}|Q_{n}(t)-Q(t)|>a\right)&\leq\frac{1}{a}\left(\|N_{n}(\infty)-N(\infty)\|_{\cL_2(\Pbb)}^2+2\|N(\infty)(N_{n}(\infty)-N(\infty))\|_{\cL_1(\Pbb)}\right) \\
&\leq \frac{4}{a}\left(\|F_{n}(\infty)-F(\infty)\|_{\cL_2(\Pbb)}^2+2\|F(\infty)\|_{\cL_2(\Pbb)}\|F_{n}(\infty)-F(\infty)\|_{\cL_2(\Pbb)}\right) \\
&\leq \frac{4}{a}\left(\|F_{n}(\infty)-F(\infty)\|_{\cL_2(\Pbb)}^2+2\|F(\infty)\|_{\cL_2(\Pbb)}\|F_{n}(\infty)-F(\infty)\|_{\cL_2(\Pbb)}\right), 
\end{flalign*}
where we've also used that $\|F(\infty)\|_{\cL_2(\Pbb)}$ is bounded by $\Ebb[(\int_0^{\infty}e^{-rt}S_0(t)|Z(t)|dt)^2]^{1/2}<\infty$ (for sufficiently big $r$). The above converges to zero as $n\rightarrow\infty$, as for sufficiently big $r$ we have $\|F_{n}(\infty)-F(\infty)\|_{\cL_2(\Pbb)}\rightarrow 0$. In fact for $\wt{Q}:=Q-Q(0)$, $\wt{Q}_{n}:=Q_{n}-Q_{n}(0)$ we have by the above that $\sup_{t\geq 0}|\wt{Q}_{n}(t)-\wt{Q}(t)|\rightarrow 0$ in probability. This yields $[\wt{Q}_{n}-\wt{Q}](\infty)\rightarrow 0$ in probability, implying $[\Gamma_{n}-\Gamma](\infty)\rightarrow 0$ in probability; which lastly gives:
\begin{equation*}
\sup_{t\geq 0}|\Gamma_{n}(t)-\Gamma(t)|\rightarrow 0,
\end{equation*}
in probability as $n\rightarrow \infty$.

Having established a concrete connection between $(\Delta_{n},\Delta)$ and $(\Gamma_{n},\Gamma)$, we move forward with explicit calculations of $\Delta$, $\Gamma$. Beginning with the former denote the measure induced by $S_0$ as $\Qbb^{\num}$ (which is locally equivalent to $\Pbb$) and note that:
\begin{equation*}
M(t)=\int_0^{t}e^{-rs}S_0(s)Z(s)ds+e^{-rt}S_0(t)\int_0^{\infty}e^{-rs}\Ebb_{z}^{\Qbb^{\num}}[Z(s)]ds.
\end{equation*}
Now, under $\Qbb^{\num}$ we have:
\begin{equation*}
dZ(t)=(\eta-kZ(t))dt+bdB^{\Qbb^{\num}}(t), \qquad \eta:=k\theta-b\si\num\rho.
\end{equation*}
Hence, we get:
\begin{equation*}
M(t)=\int_0^{t}e^{-rs}S_0(s)Z(s)ds+e^{-rt}S_0(t)\left(\frac{Z(t)+\eta/r}{r+k}\right).
\end{equation*}
Applying the Kunita-Watanabe decomposition on $M$ w.r.t. $W^1$ (which drives both $S_0$ and $S_1$), we get:
\begin{equation*}
\Delta=\int_0^{\cdot}\theta^{\Delta}(t)dW^1(t), \qquad \theta^{\Delta}(t):=\frac{(b\rho-\num\si\eta/r)}{r+k}e^{-rt}S_0(t)-\frac{\num\si}{r+k}e^{-rt}S_0(t)Z(t).
\end{equation*}
Continuing with $\Gamma$, denote the martingale closed by $[N](\infty)$ as $K$, then:
\begin{flalign*}
K(t)&=\Ebb\left[\int_0^{\infty}\frac{b^2(1-\rho^2)}{(r+k)^2}e^{-2rs}(S_0(s))^2ds\bigg|\cF(t)\right]=\int_0^{t}\frac{b^2(1-\rho^2)}{(r+k)^2}e^{-2rs}(S_0(s))^2ds\\
&+\Ebb\bigg[\int_{t}^{\infty}\frac{b^2(1-\rho^2)}{(r+k)^2}e^{-(2r-(\num\si)^2)s}\cE(2\numr)(s)ds\bigg|\cF(t)\bigg] \\
&=\int_0^{t}\frac{b^2(1-\rho^2)}{(r+k)^2}e^{-2rs}(S_0(s))^2ds+\frac{b^2(1-\rho^2)}{(r+k)^2(2r-(\num\si)^2)}e^{-(2r-(\num\si)^2)t}\cE(2\numr)(t),
\end{flalign*}
for sufficiently big $r$. Hence:
\begin{equation*}
\Gamma=\int_0^{\cdot}\theta^{\Gamma}(t)dW^1(t), \qquad \theta^{\Gamma}(t):=-\frac{2b^2(1-\rho^2)\num\si}{(r+k)^2(2r-(\num\si)^2)}e^{-(2r-(\num\si)^2)t}\cE(2\numr)(t),
\end{equation*}
where it is direct to note that if $\rho$ tends to either $1$ or $-1$, $\Gamma$ vanishes.

\section{Utility-based pricing}

The utility-based approach can also be used for the sake of pricing non-traded streams. To fix ideas we focus, for now, on the case of a finite horizon (i.e. a finite stopping time). Recalling \S 1.2, we define:
\begin{equation*}
v(x):=\sup_{\wt{X}\in\wt{\cX}(x)}\Ebb[\ln(\wt{X}(T))], 
\end{equation*}
and assume:
\begin{equation}\label{eq:ass_fifth}\tag{A5}
\text{$v(x)<\infty$, for some $x>0$,}
\end{equation}
which in turn implies that $v(x)<\infty$, for all $x>0$ by the concavity of the logarithm. In particular, the finiteness of $v(x)$ yields $(\ln(\wt{X}(T)))^{+}\in\cL_1$, for all $\wt{X}\in\wt{\cX}(x)$. To see this, note that for $\delta\in(0,1)$, $\wt{\cX}(x)\ni \wt{X}_{\delta}:=\delta+(1-\delta)\wt{X}$ is bounded away from zero. Hence $\Ebb[\ln(\wt{X}(T))]$ is well-defined with values in $\Rbb\cup\{-\infty\}$, since $\Ebb[(\ln(\wt{X}_{\delta}(T)))^{+}]\geq (1-\delta)\Ebb[(\ln(\wt{X}(T)))^{+}]$. 

In fact, $(\ln(\numX(T)))^{-}\in\cL_1$ along with $(\ln(\wt{X}(T)))^{+}\in\cL_1$, and \cite[Proposition 2.46]{KK21} imply that the wealth process in $(1,\wt{S})$ with initial value $x>0$, generated by $\theta_{i}^{\star}:=x\numX\num_{i}/\wt{S}_{i}$ for $1\leq i \leq d$, is the solution to to $v(x)$. Particularly, we get:
\begin{equation*}
v(x)=\Ebb\left[\ln\left(x+\int_0^{T}(\theta^{\star}(t))^{'}d\wt{S}(t)\right)\right]=\Ebb\left[\ln\left(x\numX(T)\right)\right].
\end{equation*}
By the above, and after recalling the discussion in Remark in \ref{rem:orig_opt}, we derive the log-based certainty equivalent $c(\epsilon)$ of the position $(1,\epsilon)$ as the solution to the following equation:
\begin{equation}\label{eq:ce_first_def}\tag{CE}
\Ebb\left[\ln\left(1+c(\epsilon)+\int_0^{T}(\theta^{\star}(t))^{'}d\wt{S}(t)\right)\right]=\Ebb\left[\ln\left(\wt{X}^{\epsilon}(T)-\epsilon\La(T)\right)\right].
\end{equation}
Rearranging and using the fact that the logarithm turns multiplication into addition, we have:
\begin{equation}\label{eq:ce_sec_def}\tag{CEII}
c(\epsilon)=e^{u(\epsilon)}-1,
\end{equation}
where $u(\epsilon)$ is given in \eqref{eq:value_discounted}. In fact, we may take \eqref{eq:ce_sec_def} as an alternative characterization of the log-based certainty equivalent, and we shall do so henceforth. The advantage is that \eqref{eq:ce_sec_def} remains valid on an infinite horizon setting, does not require \eqref{eq:ass_fifth}, and reduces to \eqref{eq:ce_first_def} whenever the more restrictive assumptions underlying it are in force. 

Putting all these together we get the following result.
\begin{corollary}
Assume the same conditions as in Theorem \ref{the:fourth_value}; then:
\begin{equation}\label{eq:fourth_ce}
\begin{aligned}
c(\epsilon)&+A_1\epsilon+\frac{A_2-A_1^2}{2}\epsilon^2+\left(\frac{A_3}{3}-\frac{A_1A_2}{2}+\frac{A_1^3}{6}\right)\epsilon^3 +\left(\frac{A_4}{4}+G-\frac{A_1A_3}{3}-\frac{A_2^2}{8}+\frac{A_1^2A_2}{4}-\frac{A_1^4}{24}\right)\epsilon^4=o(\epsilon^4),
\end{aligned}
\end{equation}
as $\epsilon\rightarrow 0+$; where:
\begin{equation*}
\begin{aligned}
A_1&:=\Ebb[F(T)], \ \ A_2:=\Ebb[(N(T))^2], \ \ A_3:=\Ebb[(N(T))^3], \ \ A_4:=\Ebb[(N(T))^4], \\
G&:=\Ebb[(\Gamma(T))^2/2-(N(T))^2\Gamma(T)].
\end{aligned}
\end{equation*}
\end{corollary}
\begin{proof}
Using \eqref{eq:fourth_value} write \footnote{In \S 1 we've shown that $u(\epsilon)<\infty$. In fact, recalling \eqref{eq:ass_illiquid_infinite} it also holds that $u(\epsilon)>-\infty$, by potentially restricting $\epsilon$ to a sub-interval of $(0,\epsilon^{L})$ s.t. $\epsilon x<1$. One way to see this is to take the positive wealth process in $S$: $X^{x}:=(1-\epsilon x)+\epsilon X$ with $X^{x}(0)=1$, where $X$ is the process implied by \eqref{eq:ass_illiquid_infinite}, starting at $x$. In particular, at maturity $X^{x}(T)-\epsilon F(T)$ is bounded below by $1-\epsilon x>0$.}:
\begin{equation*}
u(\epsilon)=u_1\epsilon+u_2\epsilon^2+u_3\epsilon^3+u_4\epsilon^4+o(\epsilon^4),
\end{equation*}
where:
\begin{equation*}
u_1:=-A_1, \ \ u_2:=-\frac{A_2}{2}, \ \ u_3:=-\frac{A_3}{3}, \ \ u_4:=-\left(\frac{A_4}{4}+G\right).
\end{equation*}
Using \eqref{eq:ce_sec_def} and expanding:
\begin{equation*}
c(\epsilon)=u(\epsilon)+\frac{(u(\epsilon))^2}{2}+\frac{(u(\epsilon))^3}{6}+\frac{(u(\epsilon))^4}{24}+o(\epsilon^4).
\end{equation*}
Collecting terms, for order $\epsilon$ we only have a contribution from $u_1$; for order $\epsilon^2$ we have contributions from $u_2$, and $u_1^2/2$; for order $\epsilon^3$ we have contributions from $u_3$, $u_1u_2$, and $u_1^3/6$; finally for order $\epsilon^4$ we have contributions from $u_4$, $u_1u_3$, $u_2^2/2$, $u_1^2u_2/2$, and $u_1^4/24$. Bringing all these together, we get the desired result.
\end{proof}
\begin{rem}
Note that in a complete market setting we have:
\begin{equation*}
\Ebb[F(T)|\cF(\cdot)]=\Ebb[F(T)]+\Delta,
\end{equation*}
i.e. the quadratic variation of $N$ as well as $\Gamma$ vanish. In turn, in this context we get:
\begin{equation*}
A_2=A_1^2, \ \ A_3=A_1^3, \ \ A_4=A_1^4, \ \ G=0.
\end{equation*}
Applying the above to \eqref{eq:fourth_ce} we have that its second, third, and fourth order terms vanish, and the log-based certainty equivalent reduces to $-\epsilon\Ebb[F(T)]$.
\end{rem}

\printbibliography
\end{document}